\documentclass[twoside,leqno,twocolumn]{article}
\pdfoutput=1 %

\usepackage[letterpaper]{geometry}
\usepackage{ltexpprt}
\usepackage[T1]{fontenc}

\usepackage[numbers]{natbib} 

\usepackage{longtable}
\usepackage[export]{adjustbox}

\usepackage{booktabs}
\usepackage[utf8]{inputenc}
\usepackage[font={small,bf},labelfont=bf,format=plain]{caption}
\usepackage{subcaption}
\usepackage{textcomp}
\usepackage{amsmath, amssymb}
\usepackage{enumitem}
\usepackage{algorithm}
\usepackage{algpseudocode}
\algtext*{EndProcedure}%
\usepackage{xcolor,colortbl}
\usepackage{footnote}
\usepackage{url}
\usepackage{tabularx}
\usepackage{multicol}
\usepackage{multirow}
\usepackage{pifont}
\usepackage{bm}
\usepackage{listings}
\usepackage{wrapfig}
\usepackage{relsize}
\usepackage{mathtools}
\usepackage{bbm}
\usepackage{blindtext}
\usepackage{bm}
\usepackage{balance}
\usepackage{mdwlist}
\usepackage{xspace}
\setlist[enumerate]{leftmargin=*}
\setlist[itemize]{leftmargin=*}

\usepackage{wasysym}
\usepackage{amssymb}%
\usepackage{pifont}%
\newcommand{\method}{PhUSION\xspace} %

\newtheorem{observation}{Observation}

\newcommand{\nonlinearity}[1]{\sigma({#1})}
\newcommand{\proximity}[1]{\Psi({#1})}
\newcommand{\embed}[1]{\zeta({#1})}
\newcommand{\aggregate}[1]{\rho({#1})}

\definecolor{mygreen}{RGB}{34,139,34}

\newcommand{\vecv}{\mathbf{v}}
\newcommand{\simMat}{\mathbf{S}}
\newcommand{\nonlinearSimMat}{\tilde{\mathbf{S}}}
\newcommand{\transMat}{\mathbf{R}}
\newcommand{\degMat}{\mathbf{D}}
\newcommand{\lapMat}{\mathbf{L}}
\newcommand{\identMat}{\mathbf{I}}
\newcommand{\adj}{\mathbf{A}}
\newcommand{\graph}{G}
\newcommand{\numberOfNodes}{n}
\newcommand{\vertexSet}{V}
\newcommand{\matU}{\mathbf{U}}
\newcommand{\matSigma}{\mathbf{\Sigma}}
\newcommand{\matLambda}{\mathbf{\Lambda}}
\newcommand{\matV}{\mathbf{V}}
\newcommand{\matZero}{\mathbf{0}}
\newcommand{\stdBasis}[1]{\mathbf{e}_{#1}}
\newcommand{\perm}{\mathbf{P}}

\newcommand{\graphVec}{\mathbf{f}}
\newcommand{\embeddingMatrix}{\mathbf{Y}}
\newcommand{\embeddingVector}{\mathbf{y}}
\newcommand{\embeddingDimension}{d}
\newcommand{\scale}{s}
\newcommand{\nScales}{t}
\newcommand{\order}{k}
\newcommand{\window}{T}

\newcommand{\heat}{\texttt{HK}\xspace}
\newcommand{\fabp}{\texttt{FaBP}\xspace}
\newcommand{\ppr}{\texttt{PPR}\xspace}
\newcommand{\ppmi}{\texttt{PPMI}\xspace}
\newcommand{\invlap}{\lapMat^{+}}
\newcommand{\adjpower}{\texttt{Adj}\xspace}
\newcommand{\rwpower}{\texttt{RW}\xspace}

\newcommand{\linear}{\texttt{Identity}\xspace}
\newcommand{\rampedLog}{\texttt{Log}\xspace}
\newcommand{\bin}[1]{\texttt{Bin-{#1}}\xspace}
\newcommand{\threshold}{p}

\definecolor{Gray}{gray}{0.85}
\definecolor{NetMF}{gray}{0.6}
\definecolor{InfWalk}{gray}{0.8}
\definecolor{HOPE}{gray}{0.92}
\newcolumntype{g}{>{\columncolor{Gray}}c} 

\begin{document}

\title{Node Proximity Is All You Need: \\
Unified Structural and Positional Node and Graph Embedding}

\author{Jing Zhu\thanks{Computer Science \& Engineering, University of Michigan. Email: \{jingzhuu, luxingyu, mheimann, dkoutra\}@umich.edu}~\thanks{Authors contributed equally to this work.}\\%
\and
Xingyu Lu\footnotemark[1]~\footnotemark[2]\\
\and
Mark Heimann\footnote{Lawrence Livermore National Laboratory. Work partially completed while a student at the University of Michigan.}%
\and
Danai Koutra\footnotemark[1] 
}

\date{}

\maketitle

\fancyfoot[R]{\scriptsize{Copyright \textcopyright\ 2021 by SIAM\\
Unauthorized reproduction of this article is prohibited}}

\begin{abstract}
    While most network embedding techniques model the relative positions of nodes in a network, recently there has been significant interest in \textit{structural embeddings} that model node \textit{role equivalences}, 
    irrespective of their distances to any specific nodes. We present \method,
    a \underline{p}roximity-based \underline{u}nified framework for computing \underline{s}tructural and posit\underline{ion}al node embeddings, which leverages well-established methods for calculating node proximity scores.
    Clarifying a point of contention in the literature, we show which step of \method produces the different kinds of embeddings and what steps can be used by both.  Moreover, 
    by aggregating the \method node embeddings, we obtain graph-level features that model information lost by previous graph feature learning and kernel methods.
    In a comprehensive empirical study with over 10 datasets, 4 tasks, and 35 methods, we systematically reveal successful design choices for node and graph-level machine learning with embeddings.%
\end{abstract}

\section{Introduction}
\label{sec:intro}
Node embeddings model node similarities in a multi-dimensional feature space:  
the more similar two nodes are in a network, the closer they lie in this space.  Two broad categories of node similarity  are prevalent in the literature:
(i)~positional proximity, which embeds close nodes similarly~\cite{gemsurvey}; and
(ii)~structural similarity, which embeds nodes similarly if they have similar roles or patterns of interaction with other nodes,
irrespective of their relative locations~\cite{graphwave}. In turn, these similarities lead to \textit{positional} or \textit{proximity-preserving} embeddings, and \textit{structural} or \textit{role-based} embeddings, respectively.

Characterizing the relationship between proximity-preserving and structural node embeddings  %
is an open and contested problem, with recent works making opposing claims. For instance,~\citeauthor{rossi2020proxstruc} characterize these classes of methods as fundamentally different methodologically and in terms of applications~\cite{rossi2020proxstruc}.  Meanwhile,  %
concurrent work proposed a \textit{theoretical} framework in which the analogous concepts are actually equivalent for downstream tasks~\cite{equiv-prox-struc}. %
However, according to~\cite{rossi2020proxstruc}, it is unclear how this theoretical framework %
maps onto real-world graph mining methods.   %

A seminal work, NetMF~\cite{netmf}, showed that various positional node embeddings amount to the same embedding technique (matrix factorization) applied to various matrices capturing pairwise node proximity scores.  Going further, we propose \method, 
a \underline{p}roximity-based \underline{u}nified framework for computing \underline{s}tructural \textit{and} posit\underline{ion}al node embeddings. 
\method has three steps: (i) computation of pairwise node proximities, (ii) application of a nonlinear filter, and (iii) application of a dimensionality-reducing embedding function. We show which steps can be used for proximity-preserving or structural embedding and which step makes them different, revealing similarities and differences between the two classes of methods.  

Additionally, \method generalizes existing methods and yields novel ones from 35 different combinations of design choices, some of which improve on the variations studied in the literature.  We extensively perform an empirical study of possible design choices for both structural and proximity-preserving node embeddings, to understand what works and why.  In particular, nonlinear filtering has very recently been identified~\cite{infinitewalk} as a key ingredient to the success of proximity-preserving node embedding.  We analyze this observation in much greater detail for proximity-preserving embeddings and for the first time apply it to structural embeddings.

We extend \method to embed entire graphs, a problem for which separate solutions have been proposed using graph signatures and similarity scores derived from node proximity matrices~\cite{netlsd, retgk} and aggregated structural node embeddings~\cite{rgm}. Since we have shown that node proximity matrices can be used to derive structural node embeddings, we interpret previous methods~\cite{netlsd,retgk} as embedding aggregation; we use \method to learn more expressive graph features by aggregating our more informative node embeddings, that model information we show that previous works cannot.    

Our contributions are summarized as follows: 
\begin{itemize*}
    \item \textbf{Unifying Embedding Perspective}: We propose \method, which can use pairwise node proximity matrix to generate embeddings that model node similarity based on structural roles or positional proximity.  Our analysis of \method shows the technical similarities and differences between structural and proximity-preserving node embeddings, a contested open question~\cite{rossi2020proxstruc, equiv-prox-struc}.
    \item \textbf{Study of Successful Design Choices}: On benchmark tasks for proximity-preserving and structural embedding choices, we investigate the combination of node proximity matrices, nonlinear transformation, and embedding functions.  Our results uncover new insights that can improve both proximity-preserving and structural embeddings.
    \item \textbf{Graph-Level Learning}: We turn \method into a method for learning features for entire networks from their node proximity matrices based on node embedding aggregation.  We interpret previous graph kernels~\cite{retgk} and feature learning methods~\cite{netlsd} as simplified versions of \method, and show what information we can capture with more expressive design choices that these previous works cannot.
\end{itemize*}
We provide code and additional supplementary material at \url{https://github.com/GemsLab/PhUSION}.

\section{Related Work}
\label{sec:related}
\noindent \textbf{Frameworks for Node Embedding.} Node embeddings are latent feature vectors for nodes in a network that are similar for similar nodes.  Most embedding methods define node similarity in terms of \textbf{proximity} (e.g. direct or indirect connection) within a single graph.  In contrast, \textbf{structural} embedding methods capture a node's structural role independent of its proximity to specific nodes; this independence makes embeddings comparable across distant parts of a graph~\cite{jin2020understanding} or separate graphs~\cite{xnetmf, rgm}.  Both kinds of embeddings may be obtained using a diverse range of shallow and deep learning methods.  For more information, we refer the reader to a survey ~\cite{gemsurvey} on proximity-preserving or positional embeddings, and a recent comprehensive empirical study on structural or role-based embeddings~\cite{jin2020understanding}.  %

The plethora of node embedding methods has raised interest in finding unifying frameworks for different methods, which can also lead to new technical advances.  For example, many proximity-preserving embedding methods %
were shown to implicitly factorize different proximity-based node similarity matrices; this insight inspired the NetMF method based on explicit matrix factorization~\cite{netmf}.  
It is known that many (proximity-preserving) node embedding methods can be summarized as a two-step process of node similarity matrix construction and dimensionality reduction~\cite{yang2017fast}.  However, \method is the first framework to subsume both proximity-preserving and structural embedding methods.  Moreover, in light of recent work~\cite{infinitewalk}, we carefully study a third step of applying a nonlinearity before performing dimensionality reduction.

\noindent \textbf{Graph Comparison.} For comparing entire graphs, aggregating node embeddings (as we do) is competitive to deep neural networks, graph kernels, and feature construction~\cite{rgm}.  Because a graph's node proximity matrix captures important information, many works have sought to use this \emph{within}-graph information for \emph{cross}-graph comparison.  A challenge is that nodes in different graphs may not correspond.  Feature learning method NetLSD~\cite{netlsd} and graph kernel RetGK~\cite{retgk} solve this problem by only considering node self-similarities, which forgoes directly modeling a node's similarity to other nodes (cross-\emph{node} similarities).  Other graph similarity functions such as DeltaCon~\cite{koutra2013deltacon} model cross-node similarities, but are restricted to graphs defined on the same set of nodes.  However, \method can model within-graph cross-node similarities for more expressive general cross-graph comparison.   

\section{{Unified Theoretical} Framework}
\label{sec:framework}
In this section, we present the abstract steps of our \method framework for node and graph feature learning, before describing concrete choices in the next section.

\vspace{0.1cm}
\noindent \textbf{Preliminaries.} We consider a graph $\graph$ with node set $\vertexSet$ and adjacency matrix $\adj$ containing edges between nodes.  We learn an $\numberOfNodes \times \embeddingDimension$ matrix $\embeddingMatrix$ of $\embeddingDimension$-dimensional node embeddings, where the $i$-th row $\embeddingMatrix_i$ is a feature representation for node $i$.  For ease of reference, we define common quantities for graph learning and node embedding, along with parameters specific to certain node proximity functions, in Tab.~\ref{tab:dfn}. 
\begin{table}[h]
\caption{Symbols and definitions}
\centering
\resizebox{\columnwidth}{!}
 {
\begin{tabular}{@{}p{1.6cm}c|l@{}}
\toprule
\multicolumn{2}{c|}{\textbf{Symbol}} & \textbf{Definition}\\\midrule
\multirow{6}{1.5cm}{Standard graph matrices}
&$\adj$&Adjacency matrix\\
&$\degMat$&Diagonal matrix of node degrees\\
&$\lapMat$&Unnormalized Laplacian matrix ($\degMat - \adj$)\\
&$\lapMat^+$&Pseudoinverse of $\lapMat$\\
&$\transMat$&Random walk transition matrix ($\degMat^{-1}\adj$)\\
&$\order$&Matrix power\\
\midrule
\multirow{6}{1.5cm}{\method functions}
&$\proximity{}$&Node proximity function\\
&$\nonlinearity{}$&Nonlinear transformation function\\
&$\embed{}$&Embedding function\\
&$\simMat$&Matrix of node proximities $\simMat = \proximity{\adj}$\\
&$\nonlinearSimMat$&Matrix of nonlinearly filtered node proximities $\nonlinearSimMat = \nonlinearity{\proximity{\adj}}$\\
&$\embeddingMatrix$&Matrix of node embeddings $\embeddingMatrix = \embed{\nonlinearity{\proximity{\adj}}}$\\
\midrule
\multirow{3}{*}{PPMI~\cite{netmf}} &vol($G$)&$\Sigma_{i,j}\adj_{ij}$\\
&T&Window size\\
&b&Parameter for negative sampling\\\midrule
\multirow{3}{1.5cm}{Heat kernel \\ \cite{graphwave,netlsd}}
&$g_\scale$&Filter kernel with scaling parameter $\scale$\\
&$\matLambda$&Diagonal matrix of eigenvalues of $\lapMat$\\
&$\matU$&Eigenvectors of $\lapMat$ ($\lapMat = \matU\matLambda \matU^T$)\\\midrule
\multirow{4}{*}{FaBP~\cite{fabp}}
&\multirow{2}{*}{$h_h$}&$\sqrt{(-c_1 + \sqrt{c_1^2 + 4c_2})/8c_2}$, \\
&& where $c_1=$trace$(\degMat) + 2$;  $c_2=$trace$(\degMat) - 1$ \\
&$a$&$4h_h^2/(1-4h_h^2)$\\
&$c$&$2h_h/(1-4h_h^2$)\\\midrule
\multirow{1}{*}{PPR~\cite{hope}}
&$\beta$&Decay parameter\\
\bottomrule
\end{tabular}
}
\label{tab:dfn}
\end{table}

\noindent \textbf{Structural vs Positional Embeddings}.  Structural node embedding should learn similar features for automorphically equivalent or near-equivalent nodes~\cite{jin2020understanding, equiv-prox-struc}, even if they are distant from each other in the network.  On the other hand, for nodes to have similar positional embeddings, they must be close in the network. Although these are two very different embedding outcomes, the steps we present below can generate either kind of embedding; later, we will show concretely where the difference arises.   

\subsection{Node Feature Learning}

For learning node features from a graph with adjacency matrix $\adj$, we perform the following three steps:

\begin{enumerate}[start=1,label={\bfseries Step~\arabic*:}]
    \item Calculate node proximities $\simMat$ using a function $\proximity{\adj}$; 
    \item Filter these proximities via a nonlinearity function $\nonlinearSimMat = \nonlinearity{\simMat}$; and 
    \item Embed the transformed proximities using a dimensionality reduction function: $\embeddingMatrix = \embed{\nonlinearSimMat}$. 
\end{enumerate}
Our node embedding framework can be precisely summarized by function composition:
\begin{equation}
\label{eq:node-embed}
    \embeddingMatrix = \embed{\nonlinearity{\proximity{\adj}}}
\end{equation}

\noindent \textbf{Multiscale Node Embeddings}.  Many proximity functions can be tuned with scaling parameters to capture more local or global proximity~\cite{graphwave,cao2015grarep}.  We can create multiscale embeddings by concatenating embeddings using the same node proximity function at several different scales:
\begin{equation}
\label{eq:node-embed-multiscale}
    \embeddingMatrix = ||_i \embeddingMatrix^{(\scale_i)} = \embeddingMatrix^{(\scale_1)} || \embeddingMatrix^{(\scale_2)} || \ldots || \embeddingMatrix^{(\scale_\nScales)},
\end{equation}
where embeddings at each individual scale are computed with Eq.~\eqref{eq:node-embed} using the desired scale parameter to compute node proximity: $\embeddingMatrix^{(\scale_i)} = \embed{\nonlinearity{\proximity{\adj; \scale_i}}}$.  %

\subsection{Graph Feature Learning}

We can aggregate a graph's node embeddings into a single feature vector that describes the entire graph using a function $\aggregate{}$: 
\begin{equation}
\label{eq:graphfeatures}
    \graphVec = \aggregate{\embeddingMatrix}
\end{equation}

\section{Unifying Node Embedding Methods}
\label{sec:node}
We now propose concrete function choices for Eqs.~\eqref{eq:node-embed}-\eqref{eq:graphfeatures}, and characterize general and specific choices.  

\subsection{Step 1: Computing Node Proximities $\proximity{}$.}
The first step of our framework, \method, is to create a matrix of pairwise node proximities $\simMat \in \mathbf{R}^{n \times n}$.  $\simMat_{ij}$ should be large for nodes that are close in the graph (e.g. neighbors) and small for faraway nodes.  Different proximity matrices have been used not only for node embedding but throughout graph mining, %
including: 
\begin{itemize*}
    \item Positive pointwise mutual information (\ppmi)~\cite{netmf}: \\
    $\simMat = \frac{\text{vol}(G)}{bT}(\sum^T_{r=1}\transMat^{r})\degMat^{-1}$.
    
    \item Heat kernel (\heat)~\cite{graphwave}:
    $\simMat = \matU g_s(\matLambda)\matU^\top$.
    
    \item Belief Propagation (\fabp)~\cite{fabp}:
    $\simMat = (\identMat + a\degMat - c\adj)^{-1}$.
    
    \item Personalized Pagerank (\ppr)~\cite{hope}: $\simMat = (\identMat - \beta\adj)^{-1}(\beta\adj)$.
    
    \item Laplacian pseudoinverse ($\invlap$)~\cite{infinitewalk}: 
    $\simMat = \lapMat^{+}$, which approximates the PPMI matrix as the window size $T \rightarrow \infty$, up to a low-rank correction term.
    
    \item Powers of the adjacency matrix (\adjpower)~\cite{hope,cao2015grarep} or random walk matrix (\rwpower)~\cite{retgk}: $\simMat = \adj^\order$ or $\simMat = \transMat^\order$.  %
\end{itemize*}

\subsection{Step 2: Nonlinear Transformations of Node Proximities $\nonlinearity{}$.}
As a preprocessing step before embedding, we can filter the node proximities with a nonlinear function $\nonlinearity{\simMat}$.  Recent work~\cite{infinitewalk} argues that such nonlinearity is largely responsible for the performance gain of recent deep learning-inspired node embedding methods.  Thus, we consider the following functions:
\begin{itemize*}
    \item No nonlinearity: $\nonlinearity{\simMat}=\simMat$ (\linear function). %
    \item Elementwise logarithm (\rampedLog): 
    For proximity-preserving embedding with \ppmi, we set 
    $\nonlinearity{\simMat}_{i,j} = \log(\max\{\simMat_{i,j},1\})$~\cite{netmf}.
    For other matrices with values concentrated in $[0, 1]$, we propose to keep more information by only filtering out negative or zero elements:
    \begin{displaymath}
    \nonlinearity{\simMat}_{i,j} = \left\{
    \begin{aligned}
    & 0 \ \ &, \simMat_{i,j} \leq 0\\
    & \log(\frac{\simMat_{i,j}}{\min(\simMat^+)}) &, \simMat_{i,j} > 0
    \end{aligned}
    \right.
    \end{displaymath}
    
    where $\min(\simMat^+)$ is the smallest positive element in $\simMat$. %
    \item Thresholded binarization (\bin{\threshold})~\cite{infinitewalk}: 
    Let $a\in\mathbf{N}$ be the $\threshold$-th percentile ($\threshold\%$ smallest element) in $\simMat$.  Then $\nonlinearity{\simMat}$ is defined elementwise as:
    \begin{displaymath}
    \nonlinearity{\simMat}_{i,j} = \left\{
    \begin{aligned}
    & 0 \ \ &, \simMat_{i,j} \leq a\\
    & 1 &, \simMat_{i,j} > a
    \end{aligned}
    \right.
    \end{displaymath}
\end{itemize*}

\subsection{Step 3: Embedding Node Proximities $\embed{}$.}
Given a (filtered) similarity matrix $\nonlinearSimMat$, node embeddings learn low-dimensional feature representations using various dimensionality reduction techniques.  We represent the embedding process as a function $\embed{\nonlinearSimMat}$.

\begin{itemize*}
        \item One way to generate $\embeddingDimension$-dimensional embeddings is by factorizing the node similarity matrix, prototypically with singular value decomposition (SVD)~\cite{netmf}.  Based on  rank-$\embeddingDimension$ SVD $\nonlinearSimMat \approx \matU_\embeddingDimension \matSigma_\embeddingDimension \matV_\embeddingDimension$, we can obtain the node embeddings as  %
$\embed{\nonlinearSimMat} = \matU_\embeddingDimension \matSigma_\embeddingDimension^{\frac{1}{2}}$.

    \item Another way to generate a $\embeddingDimension$-dimensional embeddings from an $\numberOfNodes \times \numberOfNodes$ similarity matrix $\nonlinearSimMat$ is characteristic function sampling (CFS).  
For even dimensionality $\embeddingDimension$, we compute the embedding of each node $u$ by sampling real and imaginary components from its empirical characteristic function, $\phi_u(t) = \sum_{v=1}^{\numberOfNodes}\exp(it\nonlinearSimMat_{vu})$, evaluated at $\frac{\embeddingDimension}{2}$ evenly spaced landmarks $t_1, \ldots, t_{\embeddingDimension/2}$ between 0 and 100~\cite{graphwave}. CFS is a permutation-invariant function applied row-wise to $\nonlinearSimMat$ that models the distribution of a node's proximity scores~\cite{graphwave}.
\end{itemize*}

\noindent \textbf{Special Cases.} \method generalizes several existing proximity-preserving and structural embedding methods, which we summarize in the following result:   
\begin{theorem}
\label{thm:special-cases}
Special cases of Eq.~\eqref{eq:node-embed-multiscale} include but are not limited to: GraphWave~\cite{graphwave}, NetMF~\cite{netmf}, InfiniteWalk~\cite{infinitewalk}, HOPE~\cite{hope}, GraRep~\cite{cao2015grarep},  DNGR~\cite{dngr}, and sRDE~\cite{heimann2020structural} for signed networks.
\end{theorem}
\begin{proof}
We give the constructions in App.~\ref{app:cases}.
\end{proof}

\subsection{What Makes Node Embeddings Positional or Structural?}
\label{sec:theory}
We isolate the embedding function $\embed{}$ as the responsible design choice for making \method yield positional or structural embeddings. 
Concretely, embedding a proximity matrix using SVD produces positional embeddings, while using CFS (or any other permutation-invariant row function) produces structural embeddings.  
On the other hand, any choice of $\proximity{}$ and $\nonlinearity{}$ can yield positional or structural embeddings. 

\begin{theorem}
\label{thm:embfunc}
Let connected graphs $\graph_1, \graph_2$ have an isomorphism $\pi: \vertexSet_1 \rightarrow \vertexSet_2$, i.e. a bijective mapping between the nodes and $\adj_2 = \perm \adj_1 \perm^\top $, where the binary matrix $\perm$ has nonzero elements exactly at the entries $(\pi(i), i)$ for $i \in [1, \ldots, |\vertexSet|]$.  Define a combined graph $\graph$ with block diagonal adjacency matrix $\adj = [\adj_1, \matZero; \matZero, \adj_2]$, so that $\pi$ encodes an automorphism within $\graph$.  Assume that node proximity and nonlinearity functions $\proximity{}$ and $\nonlinearity{}$ preserve this automorphism: $\nonlinearSimMat_2 = \perm \nonlinearSimMat_1 \perm^\top$, where $\nonlinearSimMat_i = \nonlinearity{\proximity{\adj_i}}$.  Also assume that disconnected nodes have proximity score $0$ (unchanged by nonlinearity), so that $\nonlinearSimMat = \nonlinearity{\proximity{\adj}} = 
[\nonlinearSimMat_1, \matZero; \matZero, \nonlinearSimMat_2]$. Let $\embeddingMatrix$ be the combined embeddings of $\graph$, which can be split into embeddings $\embeddingMatrix^{(1)}$ and $\embeddingMatrix^{(2)}$ corresponding respectively to the nodes originally in $\graph_1$ and $\graph_2$.  Then: 
\begin{enumerate}
    \item If $\embeddingMatrix = \text{SVD}(\nonlinearSimMat)$, then $\embeddingMatrix^{(1)}_i \neq \embeddingMatrix^{(2)}_{\pi(i)}$.
    \item If $\embeddingMatrix = \text{CFS}(\nonlinearSimMat)$, or more generally any permutation-invariant function $\embed{\nonlinearSimMat}$, then $\embeddingMatrix^{(1)}_i = \embeddingMatrix^{(2)}_{\pi(i)}$.
\end{enumerate}

\end{theorem}
\begin{proof}
See supplementary App.~\ref{app:embfunc}.
\end{proof}

\noindent \textbf{Note:} Some existing methods learn structural embeddings with implicit or explicit matrix factorization~\cite{struc2vec,xnetmf}, which in \method would produce positional embeddings.  The key difference is that these methods do not factorize a pairwise node proximity matrix, but a \emph{structural} similarity matrix (where disconnected nodes may have a nonzero similarity score).  One advantage of \method is that the node proximity matrices we use are well studied throughout graph mining.  %

\section{Unifying Graph Embedding Methods}
\label{sec:graph}
Our \method framework also produces features that describe an entire graph, when we aggregate its nodes' embeddings into a single feature vector.  %
Here, we show that two recent graph kernels and feature maps are in essence special cases of \method.  

\vspace{0.1cm}
\noindent \textbf{\method:NetLSD}.  NetLSD computes graph features from its heat kernel matrix at multiple scales~\cite{netlsd}.  For scales $\scale_1, \ldots, \scale_\embeddingDimension$, the resulting $\embeddingDimension$-dimensional feature vector has as its $i$-th entry $h^{(\scale_i)}$, the trace of the heat kernel matrix at scale $\scale_i$.  %
For size invariance, the authors propose normalizing an $n$-node graph's features by the heat trace of the $n$-node empty graph, which amounts to multiplying by $\frac{1}{n}$. %
Thus, the exact normalized NetLSD features are: 
$\frac{1}{n} [h^{(\scale_1)}, \ldots, h^{(\scale_\embeddingDimension)}]$.

\begin{theorem}
NetLSD (using the heat kernel with empty graph normalization) is a special case of Eq.~\eqref{eq:graphfeatures} where $\proximity{}$ computes the graph's heat kernel matrix at multiple scales $\scale$ as its proximity matrix $\simMat$, $\embed{\simMat} = \text{diag}(\simMat)$, $\nonlinearity{}$ is the identity function, and $\aggregate{}$ averages the embeddings.  %
\end{theorem}

\begin{proof}
At scale $\scale_k$, the one-dimensional node embedding of node $i$ is given by $\embeddingVector_i^{(\scale_k)} = \simMat^{(\scale_k)}_{ii}$.  Thus, for $\embeddingDimension$ scales $\scale_1, \ldots, \scale_\embeddingDimension$, the multiscale embedding of node $i$ given by Eq.~\eqref{eq:node-embed-multiscale} is $\embeddingVector_i = [\simMat^{(\scale_1)}_{ii}, \ldots, \simMat^{(\scale_\embeddingDimension)}_{ii}]$.  Aggregating these node features into graph features using Eq.~\eqref{eq:graphfeatures} gives $\graphVec = \frac{1}{n}\sum_i \embeddingVector_i = \frac{1}{n}[\sum_i \simMat^{(\scale_1)}_{ii}, \ldots, \sum_i \simMat^{(\scale_\embeddingDimension)}_{ii}] = \frac{1}{n}[\text{Tr}(\simMat^{(\scale_1)}), \ldots, \text{Tr}(\simMat^{(\scale_\embeddingDimension)})]$.  When $\simMat$ is the heat kernel matrix, each term becomes $\text{Tr}(\simMat^{(\scale_i)}) = h^{(\scale_i)}$.%
\end{proof}

\noindent \textbf{\method:RetGK}.  The scalable graph kernel ($\text{RetGK}_\text{II}$)~\cite{retgk} based on approximate feature maps~\cite{rahimi2008random} is defined as 
$K(\graph_1, \graph_2) = \mathbf{\kappa} \Big(\overline{\graphVec}(\graph_1), \overline{\graphVec}(\graph_2)\Big)$.  
Without node attributes, $\overline{\graphVec}(G) = \sum_{i=1}^{n} \phi(\embeddingVector_i)$ where the $j$-th entry of $\embeddingVector_i$ is the return probability of a random walk of length $j$ starting from node $i$ (formally $\transMat^j_{ii}$), and $\phi$ is a feature map approximating a vector-valued kernel~\cite{rahimi2008random}. It can thus be seen that RetGK has essentially the same form as the other methods:

\begin{theorem}
Without node attributes and with $\phi$ and $\kappa$ both set to the linear kernel, RetGK is a special case of Eq.~\eqref{eq:graphfeatures} where: for multiple values of the parameter $\scale$, $\proximity{}$ computes the graph's $\scale$-step random walk transition matrix as its proximity matrix $\simMat$, $\embed{\simMat} = \text{diag}(\simMat)$, $\nonlinearity{}$ is the identity function, and $\aggregate{}$ averages the embeddings.  %
\end{theorem}

In practice, ~\cite{retgk} proposes to set $\phi$ to be a random Fourier feature map to approximate the Gaussian kernel~\cite{rahimi2008random}, and $\kappa$ to be a Gaussian or Laplace kernel, applying the successive embedding trick used for graph kernels~\cite{nikolentzos2018enhancing}.  Node attributes may be incorporated by taking the Kronecker product of the attribute vectors with the embeddings~\cite{retgk}. All of these techniques readily apply to any of the other methods we have proposed.

\vspace{0.1cm}
\noindent \textbf{Expressive Graph Comparison with \method.} Postprocessing aside, we can interpret RetGK and NetLSD as instances of \method: they average multiscale embeddings learned from different node proximity matrices (\heat for NetLSD, \rwpower for RetGK).  However, they use a 1-dimensional embedding function mapping nodes to their corresponding diagonal elements in $\simMat$. Of course, this simple embedding loses off-diagonal information in $\simMat$ (namely, inter-node proximities), which our embeddings capture.  To show the greater expressivity of our embeddings $\embeddingMatrix$ by a fair comparison, we also use mean pooling for our $\aggregate{\embeddingMatrix}$, 
although more complex aggregation functions could be used~\cite{rgm}. %
 
\section{Experiments}
\label{sec:exp}
To extensively evaluate \method in a variety of contexts, we consider several real datasets  for node classification (Tab.~\ref{tab:real-node}) for which positional and structural role-based embeddings have been shown to be most effective (\S~\ref{sec:nodelevel}).  For the latter, we also use synthetic data exhibiting clear role equivalences, the structure of which we can precisely control%
~\cite{graphwave,jin2020understanding}.  We also evaluate aggregated structural embeddings for graph classification (\S~\ref{sec:graphlevel}) on real benchmark datasets
(Tab.~\ref{tab:real-graph}).  

\begin{table}[t!]
\centering
\caption{Real Datasets}
\label{tab:real-data}
\vspace{-0.2cm}
\begin{subtable}{1.0\columnwidth}
\caption{Node Classification}
{\footnotesize 
\resizebox{\columnwidth}{!}
 {
    \begin{tabular}{llrrl}
    \toprule
         & \textbf{Dataset} & \textbf{\# Nodes} & \textbf{\# Edges} & \textbf{Labels} \\ \midrule
        
        \multirow{3}{*}{\rotatebox[origin=c]{90}{Proxim.}} & \textbf{BlogCatalog}~\cite{netmf}  & 10,312 & 333,983 & Blogger Interests (39) \\
        & \textbf{PPI}~\cite{netmf} & 3,890 & 76,584 & Biological states (50) \\
        & \textbf{Wikipedia}~\cite{netmf} & 4,777 &184,812 & Part-of-Speech tags (40) \\
        \midrule
        \multirow{3}{*}{\rotatebox[origin=c]{90}{Struct.}} & \textbf{Brazil}~\cite{struc2vec} &  131 & 1,038 & \# landings \& take-off (4) \\ %
        & \textbf{Europe}~\cite{struc2vec} &  399 & 5,995 & \# landings \& take-off (4) \\
        & \textbf{USA}~\cite{struc2vec}  & 1,190 & 13,599 & \# passengers (4) \\
        \bottomrule
        
    \end{tabular}
    }
}
\vspace{-0.05cm}
\label{tab:real-node}
\end{subtable}
\vspace{0.1cm}

\begin{subtable}{1.0\columnwidth}
\caption{Graph Classification}
{\footnotesize 
\resizebox{\columnwidth}{!}
 {
    \begin{tabular}{lrrl}
    \toprule
        \textbf{Dataset} & \textbf{\# Graphs} & \textbf{Avg \# Nodes} & \textbf{Labels} \\ \midrule
        \textbf{IMDB-M}~\cite{tudataset} & 1,500 & 13.00 & Collaboration genre (3)  \\
        \textbf{PROTEINS}~\cite{tudataset} &1,113 & 39.06 & Protein type (2)  \\
        \textbf{PTC-MR}~\cite{tudataset} & 344 & 14.29 & Molecular property (2)  \\
        \bottomrule
    \end{tabular}
    }
}
\vspace{-0.05cm}
\label{tab:real-graph}
\end{subtable}

\end{table} %

\subsection{Node-level Embedding.}
\label{sec:nodelevel}
First we evaluate \method in the node classification task with positional and structural node embeddings.

\vspace{0.05cm}
\noindent \textbf{Setup.}
We combine 7 node proximity functions $\proximity{}$ and 5 different nonlinearities $\nonlinearity{}$ (including \linear).  Following our theoretical analysis (\S\ref{sec:theory}), we use SVD to generate positional node embeddings and CFS to generate structural embeddings.  In total, the \method framework gives us \textbf{35 different node embedding methods} of each type, including positional embeddings NetMF~\cite{netmf}, InfiniteWalk~\cite{infinitewalk}, and HOPE~\cite{hope} and structural embedding method GraphWave~\cite{graphwave} as special cases.  We tune hyperparameters with grid search and report the procedure and best parameters in App.~\ref{app:prox-settings}. Interestingly, we find that the best parameters strongly model local node proximity.  

We follow the supervised machine learning setup of \cite{struc2vec}: we randomly sample 80$\%$ of the dataset for training and the rest for testing. For multi-label prediction, we use the one-vs-rest logistic regression model \cite{netmf} and evaluate using micro-F1 scores. %

\subsubsection{Positional Node Embedding.} 
\label{sec:exp-node-prox}

\begin{figure*}[t]
    \centering
    \subfloat{%
      \includegraphics[width=.4\textwidth]{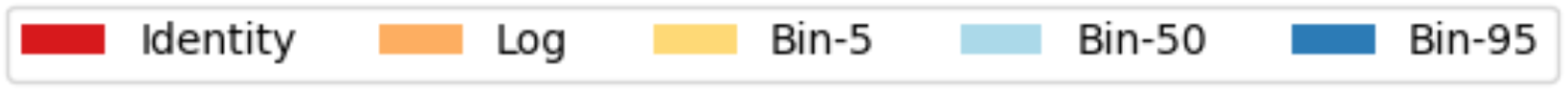}
    }
	\vspace{-0.1cm}
    
    \includegraphics[width=.3\linewidth]{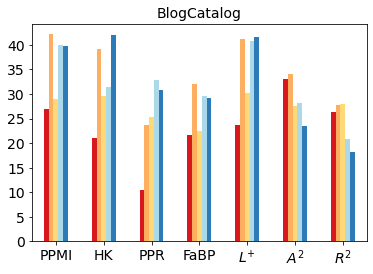}
    \includegraphics[width=.3\linewidth]{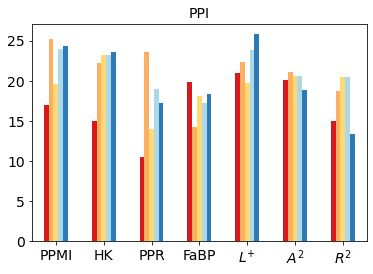}
    \includegraphics[width=.3\linewidth]{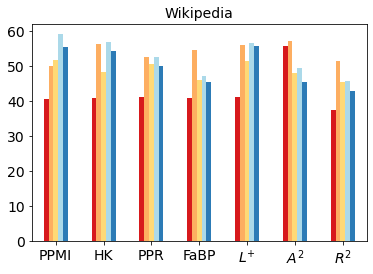}
    \caption{Node classification performance (micro-F1 scores) with positional embedding.  Nonlinearity generally helps, but the best nonlinearity function varies across proximity matrices, and the best proximity matrix varies across datasets.} %
    \label{fig:prox-classification}
\end{figure*}

We report raw results for all 35 positional node embedding methods derived from \method in Fig.~\ref{fig:prox-classification}.  Table~\ref{tab:rank-prox} performs a drilldown on a per-design choice basis.  

\vspace{0.1cm}
\noindent \textbf{Results.} 
We can see that \ppmi does an excellent job, while $\invlap$ is also competitive. As for the nonlinearity $\nonlinearity{}$, our findings support recent work~\cite{infinitewalk} that adding nonlinearity is a critical part of outperforming the original spectral embedding approaches: it is almost always beneficial for all proximity matrices.  On average, we find that \rampedLog does the best; however, \bin{\threshold} also performs better than \linear (no nonlinearity), and indeed the best embedding method for two of the three datasets (PPI and Wikipedia) uses binarization.  

\begin{table*}[t!]
\centering
\caption{Average rank and average/max micro-F1 scores of different proximity $\proximity{}$ and nonlinearity functions $\nonlinearity{}$ on all datasets used for positional node embedding.  Design choices used in existing methods \colorbox{NetMF}{NetMF} and \colorbox{InfWalk}{InfiniteWalk} perform well on average (better than \colorbox{HOPE}{HOPE}, which uses various $\proximity{}$ functions but no nonlinearity).  However, new design combinations are competitive. %
}
\label{tab:rank-prox}
\resizebox{.98\textwidth}{!}{
\begin{tabular}{@{}ll ccc c ccc c ccc@{}}
\toprule
& & \multicolumn{3}{c}{\textbf{BlogCatalog}} && \multicolumn{3}{c}{\textbf{PPI}} && \multicolumn{3}{c}{\textbf{Wikipedia}} \\ \cline{3-5} \cline{7-9} \cline{11-13}
& & \textbf{Avg Rank} & \textbf{Avg Acc} & \textbf{Max Acc} && \textbf{Avg Rank} & \textbf{Avg Acc} & \textbf{Max Acc} && \textbf{Avg Rank} & \textbf{Avg Acc} & \textbf{Max Acc}\\  \cmidrule{1-2} \cmidrule{3-5} \cmidrule{7-9} \cmidrule{11-13}
\multirow{7}{*}{$\proximity{}$} & \colorbox{NetMF}{\ppmi} & 10.4 & 35.56 & \textbf{42.21} &&  10.8 & 22.03 &25.25 && 14.2 & 51.41 & \textbf{59.10} \\
& \heat & 13.2 & 32.66 & 41.99 && 12.0 & 21.44 & 23.61 && 13.4 & 51.33 &56.89\\
& \colorbox{HOPE}{\ppr} & 21.6 & 24.61 & 32.80 && 23.8 & 16.84 &23.63 && 17.2 & 49.38 &52.69\\
& \fabp & 20.6 & 26.94 &31.98 && 24.8 & 17.57 &19.89 && 22.2 & 46.79 & 54.43\\
& \colorbox{InfWalk}{$\invlap$} & 10.0 & 35.49 & 41.55 && 8.8 & 22.52 &\textbf{25.80} && 11.8 & 52.11 & 56.47\\
& \colorbox{HOPE}{$\adj^2$} & 17.2 & 29.25&34.06 && 15.2 & 20.21&21.04 && 14.6 & 51.13&57.01\\
& $\transMat^2$ & 26.0 & 24.25&28.04 && 23.0 & 17.61&20.48 && 25.4 & 44.57&51.52\\ \cmidrule{1-2} \cmidrule{3-5} \cmidrule{7-9} \cmidrule{11-13} %
\multirow{5}{*}{$\nonlinearity{}$} & \colorbox{HOPE}{\linear} & 25.43 & 23.31 & 33.04 &&  24.0 & 16.9&20.93 && 27.86 & 42.53 &55.82\\
& \colorbox{NetMF}{\rampedLog} & 10.86 & 34.30& \textbf{42.21} && 12.86 & 21.07&25.25 && 8.86 & 53.97&57.01\\
& \colorbox{InfWalk}{\bin{5}}& 20.43 & 27.44 &30.14 && 18.71 & 19.39 &23.23 && 19.14 & 48.76&51.77\\
& \colorbox{InfWalk}{\bin{50}}& 14.0 & 31.95 &40.85 && 12.71 & 21.16 &23.97 && 11.57 & 52.53&\textbf{59.10}\\
& \colorbox{InfWalk}{\bin{95}}& 14.29 & 32.11&41.99 && 16.29 & 20.21&\textbf{25.80} && 17.43 & 49.86&55.63 \\
\bottomrule
\end{tabular}
}
\end{table*}

The use of binarization as nonlinearity and $\invlap$ for proximity was proposed by InfiniteWalk~\cite{infinitewalk}, and the use of \ppmi node proximities with \rampedLog nonlinearity is the NetMF method~\cite{netmf}.  Our findings confirm that these recently identified design choices are indeed among the most successful overall.  However, new design choices are competitive with them and may warrant further exploration.  Moreover, no single choice of nonlinearity function $\nonlinearity{}$ performs best, nor does performance vary monotonically with the sparsity of the resulting matrix (\bin{50} performs better than both \bin{5} and \bin{95}).  Corroborating ~\cite{infinitewalk}, deeper characterization of various choices of $\nonlinearity{}$ and their effects is of continued interest. 
\begin{observation}
\begin{enumerate*}
    \item[(1)] Nonlinearity has a complex effect, but is essential in improving the performance of positional node embedding.
    \item[(2)] Generally, design choices identified by recent works~\cite{netmf,infinitewalk} are among the most successful across datasets, but new combinations are often competitive.  
\end{enumerate*}
\end{observation}

\subsubsection{Structural Node Embedding.}
\label{sec:exp-node-struc}
We now evaluate the 35 methods we obtain from the \method framework for structural role-based node embedding in two major tasks, node classification and clustering.

\begin{table*}[t!]
\centering
\caption{\textit{Real data (left)}: Average rank and average/max accuracy of different proximity $\proximity{}$ and nonlinearity $\nonlinearity{}$ functions on all datasets used for structural node embedding.  
\textit{Synthetic data (right)}: Averaged clustering results for synthetic data with planted structural roles. For both tasks, we can dramatically improve on \colorbox{Gray}{GraphWave} by using a different proximity matrix and/or nonlinearity.
}
\label{tab:rank-struc}
\resizebox{.98\textwidth}{!}{
\begin{tabular}{@{}ll p{1cm}p{1cm}p{1cm} c p{1cm}p{1cm}p{1cm} c p{1cm}p{1cm}p{1cm} c p{1.1cm}p{1.1cm}p{1.1cm} @{}}
\toprule
& & \multicolumn{3}{c}{\textbf{Brazil}} && \multicolumn{3}{c}{\textbf{Europe}} && \multicolumn{3}{c}{\textbf{USA}} && \multicolumn{3}{c}{\textbf{Synthetic}} \\ \cline{3-5} \cline{7-9} \cline{11-13} \cline{15-17}
& & \textbf{Avg Rank} & \textbf{Avg Acc} & \textbf{Max Acc} && \textbf{Avg Rank} & \textbf{Avg Acc} & \textbf{Max Acc} && \textbf{Avg Rank} & \textbf{Avg Acc} & \textbf{Max Acc} && \multirow{2}{*}{\textbf{Hom}} & \multirow{2}{*}{\textbf{Comp}} & \multirow{2}{*}{\textbf{Silh}}\\  \cmidrule{1-2} \cmidrule{3-5} \cmidrule{7-9} \cmidrule{11-13} \cline{15-17}
\multirow{7}{*}{$\proximity{}$} & \ppmi & 28.00 & 37.72 &43.48 && 29.80 & 37.06 & 46.82 && 25.80 & 43.71 & 54.13 && .5283 & .5029 & .4986\\
& \colorbox{Gray}{\heat} & 6.80 & 68.75 & \textbf{72.37}  && 7.20 & 52.65 & 54.45 && 7.20 & 58.96 & \textbf{63.49} && .5951 & .5488 & .4392\\
& \ppr & 20.20 & 53.04 & 63.41 && 22.00 & 45.43 & 50.07 && 25.60 & 43.55 & 51.16 && .5951 & .5973 & \bf{.9307} \\
& \fabp & 19.80 & 52.70 & 70.15 && 23.00 & 44.94 & 49.90 && 22.00 & 47.65 & 56.92 && \bf{.7157} & .6627 & .5531\\
& $\invlap $ & 27.60 & 39.18 & 53.41 && 15.00 & 46.42 &\textbf{56.02} && 24.00 & 44.02 & 59.94 && .2071 & .1896 & .2499\\
& $\adj^2$ & 9.80 & 64.03 &71.85 && 12.80 & 50.27 & 53.97 && 9.80 & 57.67 & 59.83 && \textbf{.7156} & \bf{.6750} & .5760\\
& $\transMat^2$ & 13.40 & 63.01 &67.56 && 16.00 & 48.97 & 51.80 && 11.40 & 56.56 &58.56 && .6551 & .6071 & .4232\\ \cmidrule{1-2} \cmidrule{3-5} \cmidrule{7-9} \cmidrule{11-13} \cline{15-17}%
\multirow{5}{*}{$\nonlinearity{}$} & \colorbox{Gray}{\linear} & 14.23 & 60.33 &71.78 && 12.57 & 50.71 & \textbf{56.02} && 13.43 & 54.28 & 59.95 &&  & \multirow{5}{*}{N/A} & \\
& \rampedLog & 18.43 & 54.60 &71.85 && 21.71 & 44.19 & 53.65 && 16.14 & 52.01 & 62.73 && & &\\
& \bin{5}& 20.14 & 50.68 & 71.85 && 20.14 & 44.90 & 51.58 && 19.43 & 48.61 & 60.71 && & &\\
& \bin{50} & 11.85 & 60.78 & \textbf{72.37} && 13.14 & 49.21 & 54.68 && 17.14 & 51.94 & \textbf{63.49} && & &\\
& \bin{95}& 25.57 & 43.91 & 62.96 && 22.29 & 43.67 &54.45 && 23.86 & 44.68 & 56.97 && & &\\
\bottomrule
\end{tabular}
}
\end{table*}

\begin{figure*}[t]
    \centering
    \subfloat{%
      \includegraphics[width=.4\textwidth]{FIG/nonlinearity-legend-2.png}
    }

    \includegraphics[width=.3\linewidth]{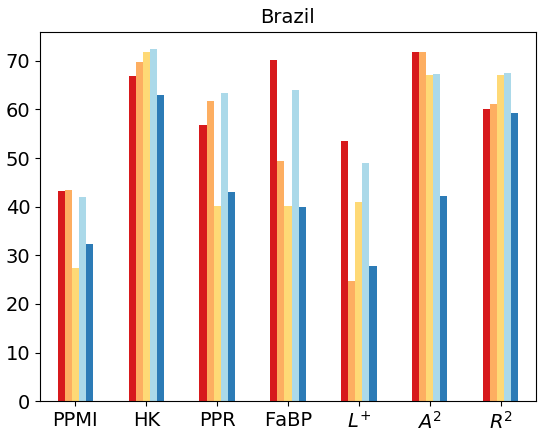}
    \includegraphics[width=.3\linewidth]{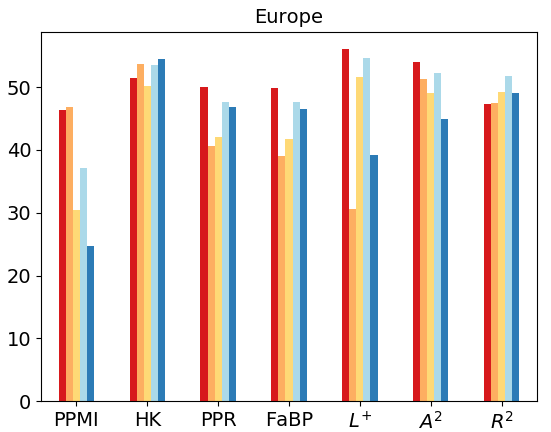}
    \includegraphics[width=.3\linewidth]{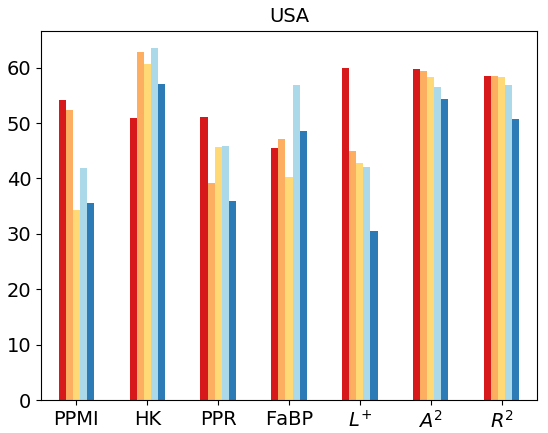}
    \caption{Node classification performance with structural embeddings. Many different proximity matrices and nonlinearity functions can yield high accuracy, often higher than existing method GraphWave.}
    \label{fig:struct-classification}
\end{figure*}

\vspace{0.1cm}
\noindent \textit{Node Classification.}
We again perform supervised machine learning to predict the node labels from the node embeddings, but in this case on datasets where the labels correspond to nodes' structural roles.  We plot the accuracy of each combination of design choice in Fig.~\ref{fig:struct-classification}, and the average rank, mean and maximum accuracy of each individual design choice in Tab.~\ref{tab:rank-struc}.

\noindent \textit{Node Clustering.}
Following the literature on structural node embedding~\cite{graphwave, jin2020understanding}, we also assess our methods using networks that are constructed to manifest distinctive structural roles.  Our goal is to cluster nodes with similar structural roles. 
We follow the dataset construction (cf. App.~\ref{app:clustering-data}) and clustering setup of \cite{graphwave}. These datasets exhibit clear role equivalence (perturbed by noise). For brevity, we only report results from embeddings without nonlinearity.  We assess the clustering quality using homogeneity, completeness, and silhouette score.

\vspace{0.1cm}
\noindent \textbf{Results}.  \textit{Node Classification.} We see different trends than positional node embeddings.  In this case, nonlinearity is not always helpful; indeed \linear is on average much more competitive.  However, all datasets, using another proximity method or nonlinearity improves on GraphWave as proposed, highlighting the flexibility of \method.  We find that a very simple nonlinearity, binarization, produces the best methods on two datasets: as CFS models the distribution of entries in each row, embedding a binary distribution simply models how many large proximities a node has to other nodes.  This corroborates a recent claim~\cite{jin2020understanding} that simple structural information suffices for these datasets.

\vspace{0.05cm}
\noindent \textit{Node Clustering.}
The results in Tab.~\ref{tab:rank-struc} (right) show that a variety of proximity matrices successfully cluster nodes by their structural roles, in some cases better than the heat kernel used in GraphWave~\cite{graphwave}.   We show similar results on unperturbed graphs in the supplementary \S~\ref{app:clustering-data}.
    
\begin{observation}
Within our \method framework, we discover design choices for structural embedding that improve on downstream tasks compared to existing methods.  In particular, we discover that some design choices used for positional node embeddings, like nonlinearity, can improve structural embeddings as well.   

\end{observation}

\begin{table*}[t!]
\centering
\caption{Graph classification using averaged node embeddings (Eq.~\ref{eq:graphfeatures}) and baselines (gray). We improve on NetLSD (3/3 datasets) and RetGK (2/3 datasets), which leverage simpler features from \heat and \rwpower matrices, using our embeddings of these matrices.  We may also use different proximity matrices like \adjpower, which can further increase performance.} %
\label{tab:graph-classification}
\resizebox{\textwidth}{!}
 {
\begin{tabular}{llllllllgg}
\toprule
  & \ppmi & \fabp & \heat & \ppr & $\adjpower$ & $\rwpower$ & $\invlap$ & NetLSD & RetGK\\  \midrule
IMDB-M  & $38.98 \pm 0.56$ & $46.54 \pm 0.27$ & $48.18 \pm 0.19$ & $41.62 \pm 0.07$  & $\mathbf{49.44 \pm 0.36}$ & $47.42 \pm 0.32$ & $45.44 \pm 0.47$ & $44.17 \pm 0.05$ & $43.91 \pm 0.74$ \\
PROTEINS & $70.50 \pm 0.36$ & $73.38 \pm 0.19$ & $73.94 \pm 0.16$ & $71.64 \pm 0.08 $ & $72.36 \pm 0.34$ & $71.52 \pm 0.17$ & $70.76 \pm 0.30$ & $71.96 \pm 0.04$ & $\mathbf{74.37} \pm 0.06$ \\
PTC-MR & $56.80 \pm 0.40$ & $55.48 \pm 0.80$ & $\mathbf{59.18 \pm 0.97}$ & $58.84 \pm 0.71$ & $55.02 \pm 0.77$ & $58.22 \pm 0.60$ & $58.44 \pm 0.59$ & $58.84 \pm 1.37$ & $57.56 \pm 1.27$ \\\bottomrule

\end{tabular}
}
\end{table*}

\subsubsection{Comparing Design Choices for Positional \& Structural Embeddings.}  Based on all our node-level experiments, we see that although the same design choices prior to embedding ($\proximity{}, \nonlinearity{}$) can be used for positional or structural embeddings, in practice the best design choices for each kind of embedding tend to be different.  For instance, nonlinearity is almost always helpful for positional node embeddings, but only sometimes helpful for structural embeddings. Proximity functions \ppmi and $\invlap$ tend to be successful for positional node embeddings, but do not produce the best structural embeddings (clearly seen on the clustering tasks).  %

This analysis raises an important question: Can we characterize node proximity matrices that produce good embeddings of either type? We perform initial exploratory analysis in App.~\ref{app:exploratory}, investigating properties of the matrices produced by each combination of $\proximity{}$ and $\nonlinearity{}$.  We find that the row-wise sums of elements in matrices producing good positional node embeddings tend to have a bell-shaped distribution, whereas we observe power-law distributions in matrices that produce good structural embeddings.%

\begin{observation}
While positional and structural node embeddings may begin with the same node proximity designs, in practice the best designs for each kind of embedding method tend to differ.
\end{observation}

This may be one reason why the survey work~\cite{rossi2020proxstruc}, characterizing existing examples of positional and structural node embedding methods, judged their methodology to be fundamentally different (even though our framework and the theory of ~\cite{equiv-prox-struc} show a methodological connection in principle).

\subsection{Graph-Level Embedding.}
\label{sec:exp-graph}
\label{sec:graphlevel}

We now investigate \method's effectiveness in learning graph features from various node proximity matrices. Intuitively, we expect that our more expressive features will allow us to classify graphs more accurately than previous works.  

\vspace{0.1cm}
\noindent \textbf{Setup.} Our experiments evaluate the graph classification accuracy on PTC-MR, IMDB-M and PROTEINS datasets~\cite{tudataset}.  As our focus is learning from the graph structure alone, we ignore node attributes. %
We only use CFS (i.e. structural embeddings), which are comparable across graphs~\cite{rgm}, and do not use nonlinearity $\nonlinearity{}$ as the baselines do not.  We use a linear SVM to predict graphs' labels from their features; we report the 10-fold cross-validation accuracy averaged over 5 trials~\cite{rgm}.  

We compare against NetLSD~\cite{netlsd} and RetGK~\cite{retgk}, alternative ways of deriving graph features from \heat and \rwpower proximity matrices, respectively (\S\ref{sec:graph}).  We use NetLSD's default 250 heat kernel values logarithmically spaced in the range $\{10^{-2}, 10^2\}$.  We run RetGK using its defaults of 50th-order random walk return probabilities and its proposed exact and approximate successive kernel embedding ($\kappa$ and $\phi$ in \S\ref{sec:graph}). %
We describe our hyperparameter settings in supplementary App.~\ref{app:prox-settings}; we parallel the settings of NetLSD and RetGK, and carefully avoid giving ourselves any unfair advantage over them (in fact, they have a slight advantage if anything: we leave NetLSD with its default higher dimensionality and RetGK with its default successive kernel embeddings). 

\vspace{0.1cm}
\noindent \textbf{Results.}
In Tab.~\ref{tab:graph-classification}, we see that our methods generally improve on NetLSD and RetGK as a way of getting graph features from their node proximity matrices. In particular, embedding \rwpower using Eq.~\ref{eq:node-embed-multiscale} outperforms RetGK, which is also based on the \rwpower proximity matrix, on two out of three datasets (PTC-MR and IMDB-M).  Similarly embedding \heat outperforms NetLSD, which also uses the heat kernel matrix, on all three datasets (and outperforms all other methods on two datasets).  This is %
strong evidence that by modeling each node's full distribution of proximities rather than its self-proximity, \method captures more useful information.  

Because we keep the embedding dimension the same as (or lower) than NetLSD and RetGK, which capture only a single value for a node at each proximity scale (whereas we return a 10-dimensional embedding), we necessarily consider much fewer scales.  Our good comparative performance indicates that modeling more graph information at fewer scales is generally superior to modeling less information at more scales.  
\begin{observation}
\method gives us a way to learn graph features from a given node proximity matrix that yield greater accuracy than previous works~\cite{netlsd,retgk}, likely because of their expressivity (\S~\ref{sec:graph}).
\end{observation}

\subsection{Additional Analysis.}
For all our classification tasks, we also study the effect of proximity order for multiscale embeddings in the supplementary App.~\ref{app:sensitivity-prox-order}.  In general, we find that modeling strongly local information with low-order proximity yields good performance (and is computationally cheapest).  %

\section{Conclusion}
\label{sec:conclusions}
We have proposed the first unifying perspective that encompasses both proximity-preserving and structural node embedding methods, clarifying their contested technical relationship~\cite{equiv-prox-struc, rossi2020proxstruc}.  This allows us to learn either kind of node embedding from any node proximity matrix that can be computed on a graph, which arises throughout the field of graph mining.  Our three-step framework \method opens up a variety of design choices (we empirically study 35), encompassing existing methods and also producing novel ones.  We provide insights into productive design choices for node-level graph mining using either kind of embedding.  By aggregating a graph's embeddings, we can derive graph-level features from the node proximities; we show precisely what information we can capture that is lost by other graph kernels and feature learning methods. 

Within \method there is still room to explore more design choices, such as other embedding functions (e.g. nonlinear autoencoders used by a few methods for positional node embedding~\cite{dngr}, or trainable characteristic function sampling recently proposed for node and graph embedding~\cite{feather}).  For graph embedding, other designs use successive kernel embedding and the incorporation of node attributes~\cite{retgk}. Furthermore, fast approximate computation of node proximities can allow \method to scale to very large graphs~\cite{netmf, graphwave}.  
\section*{Acknowledgements}
This work is supported by
NSF Grant No.\ IIS 1845491, Army Young Investigator Award No.\ W9-11NF1810397, 
and Adobe, Amazon, Facebook, and Google faculty awards. 
Any opinions, findings, and conclusions or recommendations 
expressed in this material are those of the {authors} and do not necessarily reflect the views of the %
funding parties.

\setlength{\bibsep}{0.85pt plus 0.3ex}

\bibliographystyle{unsrtnat}
\bibliography{BIB/bibliography}

\clearpage

\appendix
\section{Proofs}
\label{app:proofs}

\subsection{Existing Node Embedding Methods as Special Cases of Eq.~(\ref{eq:node-embed-multiscale}).}
\label{app:cases}
For all the methods in Theorem~\ref{thm:special-cases}, we list the specific choices of node proximity $\proximity{}$, nonlinearity $\nonlinearity{}$, and embedding $\embed{}$ functions (as well as whether or not they use multiscale proximity) that make them conform to our framework.
\begin{itemize}
\item GraphWave~\cite{graphwave}: the node proximity $\proximity{}$ computes the graph's heat kernel matrix, the nonlinearity $\nonlinearity{}$ is the identity function, and the embedding function $\embed{}$ is characteristic function sampling. The multiscale version of GraphWave is given by Equation~\ref{eq:node-embed-multiscale}.  
\item NetMF~\cite{netmf}: the node proximity $\proximity{}$ computes the graph's PPMI matrix, the nonlinearity $\nonlinearity{}$ is \rampedLog, and the embedding function $\embed{}$ is SVD.   
\item InfiniteWalk~\cite{infinitewalk}: the node proximity $\proximity{}$ computes the PPMI matrix in the window size limit $T = \infty$, or the Laplacian pseudoinverse $\invlap$ as an approximation of this quantity up to a low-rank correction term.  The nonlinearity $\nonlinearity{}$ is \rampedLog (or, for the Laplacian pseudoinverse, the authors consider \bin{\threshold}), and the embedding function $\embed{}$ is SVD.   
\item HOPE~\cite{hope}: the node proximity $\proximity{}$ computes the personalized pagerank matrix or the common neighbors matrix $\adj^2$, the nonlinearity $\nonlinearity{}$ is the identity function, and the embedding is SVD (possibly approximated for scalability~\cite{hope}).
\item GraRep~\cite{cao2015grarep}: the node proximity $\proximity{}$ is derived from powers of the adjacency matrix, the nonlinearity $\nonlinearity{}$ is $\rampedLog$, and the embedding function is SVD; this method computes multiscale node embeddings by concatenating embeddings derived from different powers of the adjacency matrix.  
\item DNGR~\cite{dngr}: the node proximity $\proximity{}$ computes the graph's PPMI matrix (in a slightly different way than NetMF), the nonlinearity $\nonlinearity{}$ is \rampedLog, and the nonlinear embedding function $\embed{}$ is implemented with a stacked denoising autoencoder. 
\item sRDE~\cite{heimann2020structural}: the node proximity $\proximity{}$ in a \emph{signed} network is computed using a signed random walk with restart procedure, the nonlinearity $\nonlinearity{}$ is the identity function, and the embedding function $\embed{}$ consists of computing a histogram (which is also permutation-invariant) of each node's signed proximity scores. 
\end{itemize}

\subsection{Embedding Functions that Produce Positional vs. Structural Node Embeddings}
Here we give the proof of Theorem~\ref{thm:embfunc}:
\label{app:embfunc}
\begin{proof}
\noindent \textbf{Part 1: SVD yields different embeddings for automorphic nodes.}
Recall that finding the SVD of $\nonlinearSimMat = [\nonlinearSimMat_1, \matZero; \matZero, \nonlinearSimMat_2]$ is equivalent to finding the eigendecomposition of $\nonlinearSimMat \nonlinearSimMat^\top$: the singular vectors (columns of $\matU$) are the eigenvectors and the singular values (diagonal entries of $\matSigma$) the square roots of eigenvalues of $\nonlinearSimMat \nonlinearSimMat^\top$.  Since the embeddings are formed from the first $\embeddingDimension$ columns of $\matU$ and $\matSigma$, we equivalently analyze the eigendecomposition of $\nonlinearSimMat \nonlinearSimMat^\top$.
\begin{enumerate}
    \item $\nonlinearSimMat_1$ and $\nonlinearSimMat_2$ are similar matrices and thus have the same eigenvalues and eigenvectors.  
    \item $\nonlinearSimMat\nonlinearSimMat^\top$ has the same eigenvalues as $\nonlinearSimMat_1$ (equivalently, $\nonlinearSimMat_2$).  First, we show that all eigenvalues of $\nonlinearSimMat_1, \nonlinearSimMat_2$ are eigenvalues of $\nonlinearSimMat \nonlinearSimMat^\top$: if $\nonlinearSimMat_1\vecv_{\lambda} = \lambda \vecv_{\lambda}$, then $\nonlinearSimMat\nonlinearSimMat^\top[\vecv_{\lambda}, \matZero] = \lambda[\vecv_{\lambda},\matZero]$; $\nonlinearSimMat\nonlinearSimMat^\top[\matZero, \vecv_{\lambda}] = \lambda[\matZero,\vecv_{\lambda}]$.  Conversely, we also show that all eigenvalues of $\nonlinearSimMat \nonlinearSimMat^\top$ are eigenvalues of $\nonlinearSimMat_1, \nonlinearSimMat_2$.  Without loss of generality we can write any eigenvector $\vecv$ of $\nonlinearSimMat$ split in half as $[\vecv_1, \vecv_2]$, such that $\nonlinearSimMat\nonlinearSimMat^\top [\vecv_1, \vecv_2] = \lambda [\vecv_1, \vecv_2]$. Then $\nonlinearSimMat\nonlinearSimMat^\top [\vecv_1, \vecv_2] = [\nonlinearSimMat_1 \nonlinearSimMat_1^\top, \matZero; \matZero, \nonlinearSimMat_2 \nonlinearSimMat_2^\top] [\vecv_1, \vecv_2] = [\nonlinearSimMat_1 \vecv_1 + \matZero \vecv_2, \matZero \vecv_1 + \nonlinearSimMat_2 \vecv_2] = [\nonlinearSimMat_1 \vecv_1, \nonlinearSimMat_2 \vecv_2]$.  Since $[\vecv_1, \vecv_2]$ was an eigenvector of $\nonlinearSimMat \nonlinearSimMat^\top$, $[\nonlinearSimMat_1 \vecv_1, \nonlinearSimMat_2 \vecv_2] = \lambda [\vecv_1, \vecv_2]$ and thus $\nonlinearSimMat_1 \vecv_1 = \lambda \vecv_1$ and $\nonlinearSimMat_2 = \vecv_2$, meaning that $\lambda$ is also an eigenvalue of $\nonlinearSimMat_1$ and $\nonlinearSimMat_2$. 
    
    \item Thus, each of the top singular vectors of $\nonlinearSimMat$ that form the dimensions of $\embeddingMatrix$ which form the embedding dimensions up to weighing by the singular values, has the form $[\matZero, \vecv_\lambda]$ or $[\vecv_\lambda, \matZero]$.  (Since the graphs are connected i.e. nonempty, $\vecv_\lambda \neq \matZero$.) That is, along any dimension the nodes in one graph will have a nonzero embedding value and the nodes in the other graph will have a zero embedding value.  
\end{enumerate}

This is of course an extreme case for a highly contrived example (perfectly automorphic nodes in perfectly disconnected components of a graph), but in general we can see (and the research community has found experimentally on real-world networks) that the SVD embeddings encode positional rather than structural information, and nodes in very different parts of the graph will generally not be close in the embedding space. 

\noindent \textbf{Part 2: Permutation-invariant row functions such as CFS yield identical embeddings for automorphic nodes.}
Let $n$ be the number of nodes in either graph $\graph_1$ or $\graph_2$.  Then the first $n$ nodes in $\graph$ correspond to $\graph_1$ and the second $n$ nodes in graph correspond to $\graph_2$. So for node $i \in [1,\ldots,n]$, the ID of its counterpart under the isomorphism $\pi$ is $\pi(i) + n$. Thus, we want to show that the rows of node $i$ and node $\pi(i) + n$ in $\nonlinearSimMat$ are equivalent up to permutation.  Formally, we show that for any $i,j \in [1, \ldots, n]$, $\nonlinearSimMat_{ij} = \nonlinearSimMat_{\pi(i) + n, \pi(j) + n}$.  

Let $\stdBasis{i}$ be the $i$-th standard basis. Then $\nonlinearSimMat_{ij} = \nonlinearSimMat_{1_{ij}} =  \stdBasis{i} \nonlinearSimMat_1 \stdBasis{j}^\top = \stdBasis{i} \perm^\top \nonlinearSimMat_2 \perm \stdBasis{j}^\top = \stdBasis{\pi(i)} \nonlinearSimMat_2 \stdBasis{\pi(j)}^\top = \nonlinearSimMat_{2_{\pi(i)\pi(j)}} = \nonlinearSimMat_{\pi(i)+n, \pi(j)+n}$. This shows that any nonzero element in the $i$-th row of $\nonlinearSimMat$ (which must occur in the first $n$ elements) has a corresponding element among the second $j$ elements of the $(\pi(i) + n)$-th row.  Of course, the second $n$ elements in the $i$-th row and the first $n$ elements in the $(\pi(i) + n)$-th row of $\nonlinearSimMat$ are zeros.  Thus, these rows have the same elements and are identical up to permutation.
\end{proof} 

\section{Node Proximity Hyperparameters}
\label{app:prox-settings}

\noindent \textbf{For positional node embeddings}:  All embeddings have the standard 128 dimensions~\cite{netmf}.  We tuned the hyperparameters of the node proximity functions \ppmi, \ppr, \heat, and \fabp on the Wikipedia dataset via grid search over the following values: 
\begin{enumerate}
    \item \heat: we tried scale values of $\scale \in [0.01,0.1,1,10,25,50]$, and find best performance from $\scale = 0.1$.
    \item \ppr: we tried decay parameter values $\beta \in [0.9, 0.5, 0.1, 0.01]$, and find best performance from $\beta = 0.01$. 
    \item \ppmi: we tried window size $\window \in [2,5,10]$ and found little difference, so we use $T=10$ with the approximate NetMF method~\cite{netmf}.
    \item \fabp: we tried  values for the parameters $a,c \in \{0.01, 0.1, 1, 10\}$.  We found little difference for values of $a$, but smaller $c$ can lead to better performance, so we chose $a=1$ and $c=0.01$. 
\end{enumerate}

\noindent \textbf{For structural embeddings}: On these smaller graphs, all embeddings are 50-dimensional. %
We tuned the hyperparameters of the node proximity functions \ppmi, \ppr, \heat, and \fabp on the USA dataset via grid search over the following values: 
\begin{enumerate}
    \item \heat: we used multiscale embeddings following~\cite{graphwave}. We found that on the airports datasets, their automatic scale selection procedure yielded unintuitively large and poorly performing scales.~\footnote{For example, applying the official implemention of GraphWave~\cite{graphwave} using the automatic scale selection on USA-airports dataset gives a range of scale parameters $\scale_{min} = 2014340.3$ and $\scale_{max} = 8763076.3$.}  Thus, we tried $\{1, 5, 10, 25, 50\}, \{0.1, 1, 10, 25, 50\}$ and $\{0.01, 0.1, 1, 10, 100\}$, and find best performance from the latter. 
    
    \item \ppr: we tried decay parameter values $\beta \in [0.9, 0.5, 0.1, 0.01]$, and find $\beta = 0.01$ works best. 
    \item \ppmi: we tried window size $\window \in [2,5,10]$ and found that $\window = 10$ achieves best performance.
    \item \fabp: we tried parameter values $a,c \in [0.01, 0.1, 1, 10]$, but in the end we found that the heuristic proposed in \cite{fabp} for setting $a$ and $c$ works best: $ a=4h_h^2/(1-4h_h^2), c=2h_h/(1-4h_h^2)$. Here the ``about-half'' homophily factor $h_h = \sqrt{\frac{-c_1+\sqrt{c_1^2 + 4c_2}}{8c_2}}$ where $c_1 = \text{Tr}(\degMat)+2, c_2 = \text{Tr}(\degMat^2)-1$.
\end{enumerate}

\noindent \textbf{For graph classification}: %
For \heat, we use scale parameters  $\scale \in \{0.01, 0.1, 1, 10, 100\}$ to parallel NetLSD. For proximity functions computed by matrix powers (\adjpower and \rwpower), we consider powers $\order \in \{1,2,3,4,5\}$.  At each of the five parameter settings, we learn 10-dimensional embeddings and use Eq.~\ref{eq:node-embed-multiscale} to form a multiscale embedding with 50 dimensions (to match or stay below the modeling capacity of NetLSD and RetGK).  Between NetLSD's higher (250) dimension and RetGK's successive kernel embeddings, our experimental setup gives NetLSD and RetGK each a small advantage.%

\section{Clustering Structural Node Embedding: Additional Details and Results} \label{app:clustering-data}
We use the synthetic graph generation pipeline provided by GraphWave~\cite{graphwave}. The graphs are given by 5 basic shapes of one of different types (``house'', ``fan'', ``star'')~\cite{graphwave} that are placed on a cycle of length 30.  In the main paper, we add 10\% random edges to perturb the otherwise perfect role equivalences of nodes in the same part of different shape; however, in Tab.~\ref{tab:clustering-unperturbed}, we include clustering results on noiseless networks exhibiting perfect role equivalence.  We use agglomerative clustering (with single linkage) to cluster the node embeddings learned by each method.

\begin{table}[h]
\centering
\resizebox{\columnwidth}{!}
 {
\begin{tabular}{@{}lllgllll@{}}
\toprule
     & \ppmi & \fabp & \heat & \ppr & $\adj^2$ & $\transMat^2$ & $\invlap$ \\  \midrule
Homogeneity  & 0.8738  & \textbf{1.000} & 0.9727 & \textbf{1.000} & \textbf{1.000} & 0.9297 & 0.8370\\
Completeness   & 0.8367 & \textbf{1.000} & 0.9407 & \textbf{1.000} & \textbf{1.000} & 0.9057  & 0.7812\\
Silhouette & 0.8241 & 0.9089 & 0.8814 & \textbf{0.9702} & 0.9204 & 0.8616 & 0.8313 \\ \bottomrule

\end{tabular}
}
\caption{Clustering results on \emph{noiseless} synthetic datasets. \heat used in \colorbox{Gray}{GraphWave} is outperformed on all metrics by other proximity matrices.}
\label{tab:clustering-unperturbed}
\end{table}

\section{Proximity Order in Structural Embedding}
\label{app:sensitivity-prox-order}
The order of proximity that node embeddings model has been shown to be very important.  While some methods by default model low-order (e.g. 2nd-order) proximities~\cite{jin2020understanding}, other methods try to balance low-order and high-order proximities to capture local and global information.  This has been done with multiscale embeddings, whether positional~\cite{cao2015grarep} or structural~\cite{graphwave}.  Setting hyperparameters that govern the order of proximity is thus important to understand.  

\noindent \textbf{Setup.}
For node and graph classification, we consider the effect of varying the order for methods consisting of powers of a (filtered) similarity matrix $\nonlinearSimMat$ (i.e. \adjpower or \rwpower) when computing multiscale embeddings (Eq.~\ref{eq:node-embed-multiscale}).  We only consider up to 4th order for positional node embeddings due to the larger size of those datasets (which is also why we omit the largest dataset BlogCatalog).  For any value of $k$, we compute the embeddings from each power of $\nonlinearSimMat, \nonlinearSimMat^2, \ldots, \nonlinearSimMat^k$ (using \rampedLog nonlinearity for positional node embeddings and \linear for structural embeddings, nonlinearity functions which performed well on average for each kind of embedding in \S~\ref{sec:exp}) and concatenate the resulting embeddings.  Note that for graph classification experiments, we now learn a 50-dimensional embedding at \emph{each} scale, as we are not comparing to baseline methods  now.  

\noindent \textbf{Results}.  
The results are shown in Fig.~\ref{fig:sensitivity-order}.  We can see, confirming the intuition of prior structural embedding methods~\cite{jin2020understanding} that lower order proximity is sufficient for best performance and saves the computational expense of computing higher order node proximities (which amounts to additional multiplications of increasingly dense matrices).  

\begin{observation}
Modeling low-order node proximity (however, beyond first-order proximity, or direct edge connections alone) is generally sufficient for both kinds of embedding methods.
\end{observation}

\begin{figure}[ht!]
\centering
    \begin{subfigure}{0.46\linewidth}
    \centering
    \includegraphics[width=\columnwidth]{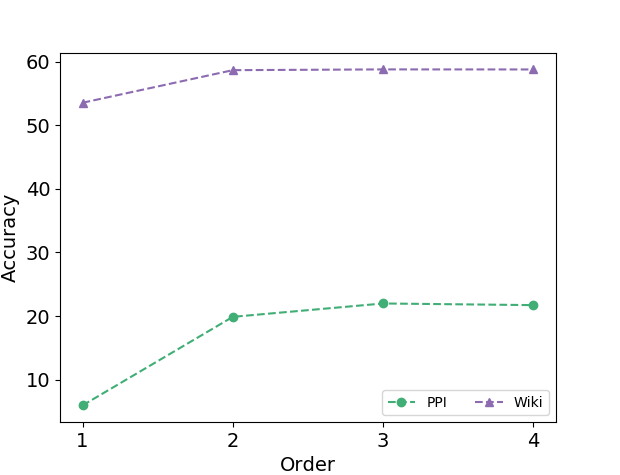}
    \caption{Proximity node classification: Multiscale \adjpower}
    \label{fig:adj-order-node}
    \end{subfigure}
    ~
    \begin{subfigure}{0.46\linewidth}
    \centering
    \includegraphics[width=\columnwidth]{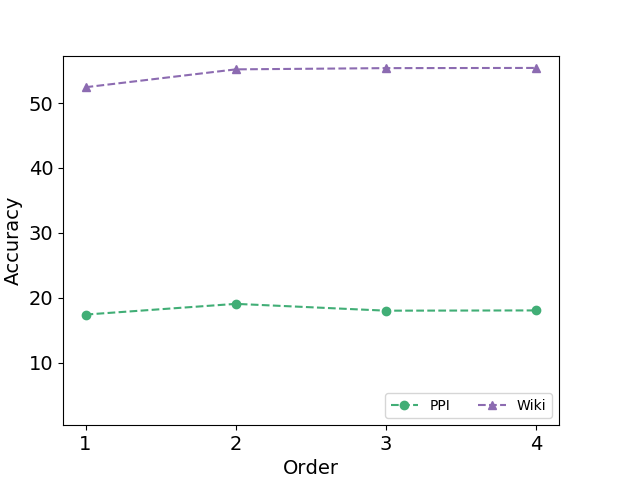}
    \caption{Proximity node classification: Multiscale \rwpower}
    \label{fig:adj-order-node}
    \end{subfigure}
    ~
    \begin{subfigure}{0.46\linewidth}
    \centering
    \includegraphics[width=\columnwidth]{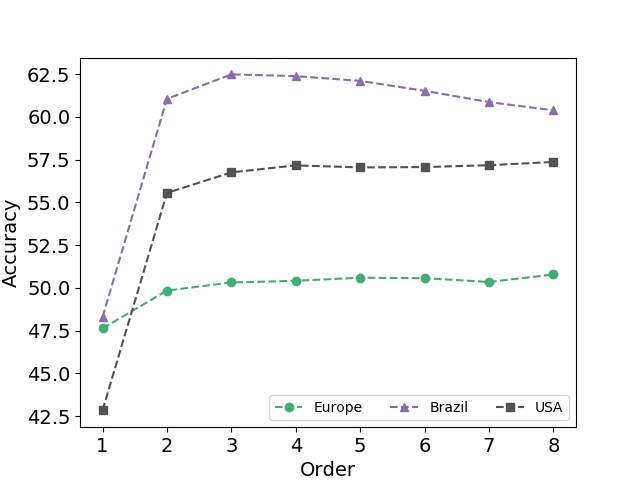}
    \caption{Structural node classification: Multiscale \adjpower}
    \label{fig:adj-order-node}
    \end{subfigure}
    ~
    \begin{subfigure}{0.46\linewidth}
    \centering
    \includegraphics[width=\columnwidth]{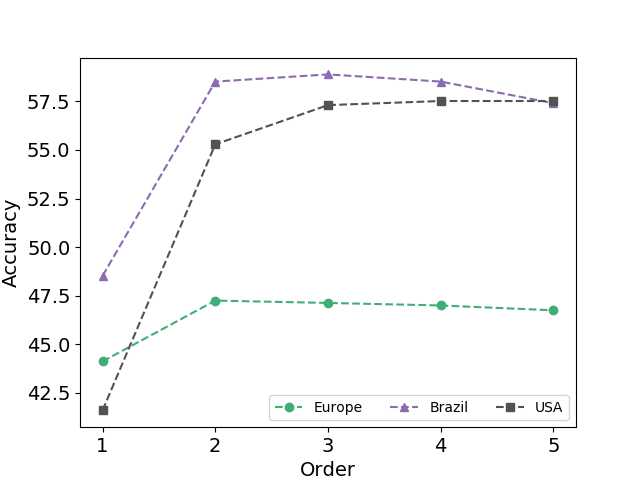}
    \caption{Structural node classification: Multiscale \rwpower}
    \label{fig:rw-order-node}
    \end{subfigure}
    
    \begin{subfigure}{0.46\linewidth}
    \centering
    \includegraphics[width=\columnwidth]{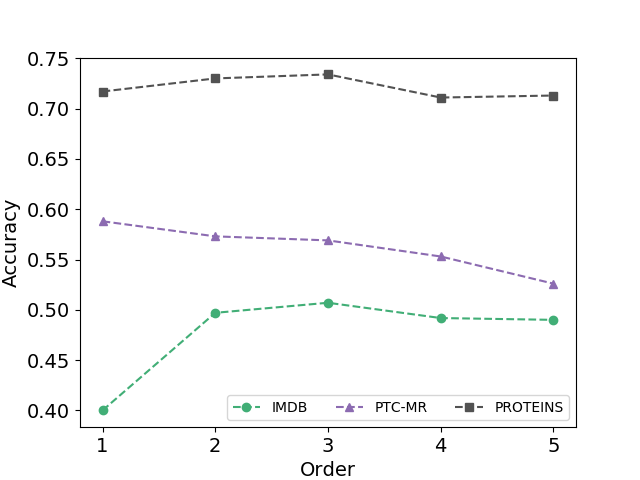}
    \caption{Graph classification: Multiscale \adjpower}
    \label{fig:adj-order-graph}
    \end{subfigure}
    ~
    \begin{subfigure}{0.46\linewidth}
    \centering
    \includegraphics[width=\columnwidth]{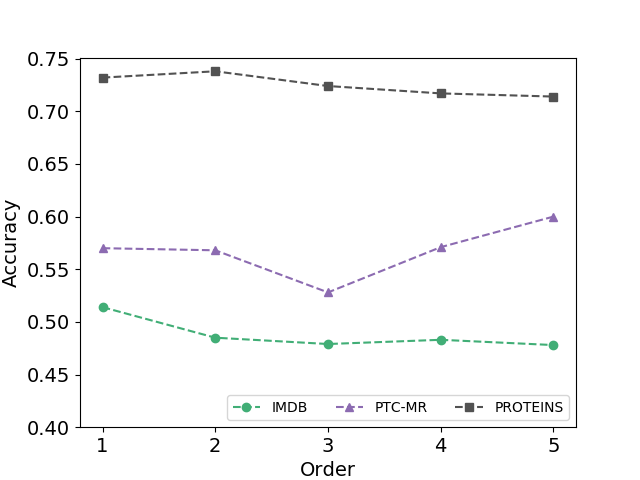}
    \caption{Graph classification: Multiscale \rwpower}
    \label{fig:rw-order-graph}
    \end{subfigure}
    
\caption{Effect of proximity order on node and graph classification.  Low order proximities are sufficient to achieve good performance.}
\label{fig:sensitivity-order}
\end{figure}

\section{Proximity Matrix Properties for Effective Node Embeddings}
\label{app:exploratory}
A powerful tool for the design of future node embedding methods would be an \emph{intrinsic} characterization of successful design choices for node embedding; this could allow for effective model selection without relying on \emph{extrinsic} evaluation (i.e. performance on downstream tasks as in \S~\ref{sec:exp}).  The node embedding step usually leverages standard dimensionality reduction techniques; from a graph mining perspective, the most interesting part is the construction of (potentially nonlinearly transformed) node proximities.   Thus, we seek to understand: how can we characterize choices $\proximity{\nonlinearity{\adj}}$ that yield useful (positional or structural) node embeddings?  While effective intrinsic analysis of node embedding methods is a major open question, we present some initial exploratory analysis to prompt further investigation.  

\subsection{Positional Embeddings}
\label{app:prox-stat-dist}
In a node proximity matrix, the sums of each row correspond to the total proximity scores each node has to all other nodes.  Our intuition is that if we are to expect good positional node embeddings, most nodes should have a moderate amount of total proximity to other nodes---too low, and the embedding objective will have too little similarity information to learn an effective embedding; too high, and the embedding objective will try to embed this node indiscriminately similarly to many other nodes.

\noindent \textbf{Setup.}
In Fig.~\ref{fig:prox-degdist}, we visualize the distribution of row sums of all node proximity matrices, arising from each combination of node proximity and nonlinearity function that we evaluated in this work.

\noindent \textbf{Results.}
Some of these distributions exhibit a bell curve shape with the values concentrated in the middle of the distribution, while others exhibit a power law distribution with a single long tail.  (Note that for the \bin{5} nonlinearity, the tail is on the left as most values in the matrix are 1, so low row sums are the exception.  In general, the tail consists of the large row sums, as is typical for most power law distributions.)  

\begin{figure*}[t]
    \centering
    \includegraphics[width=\linewidth]{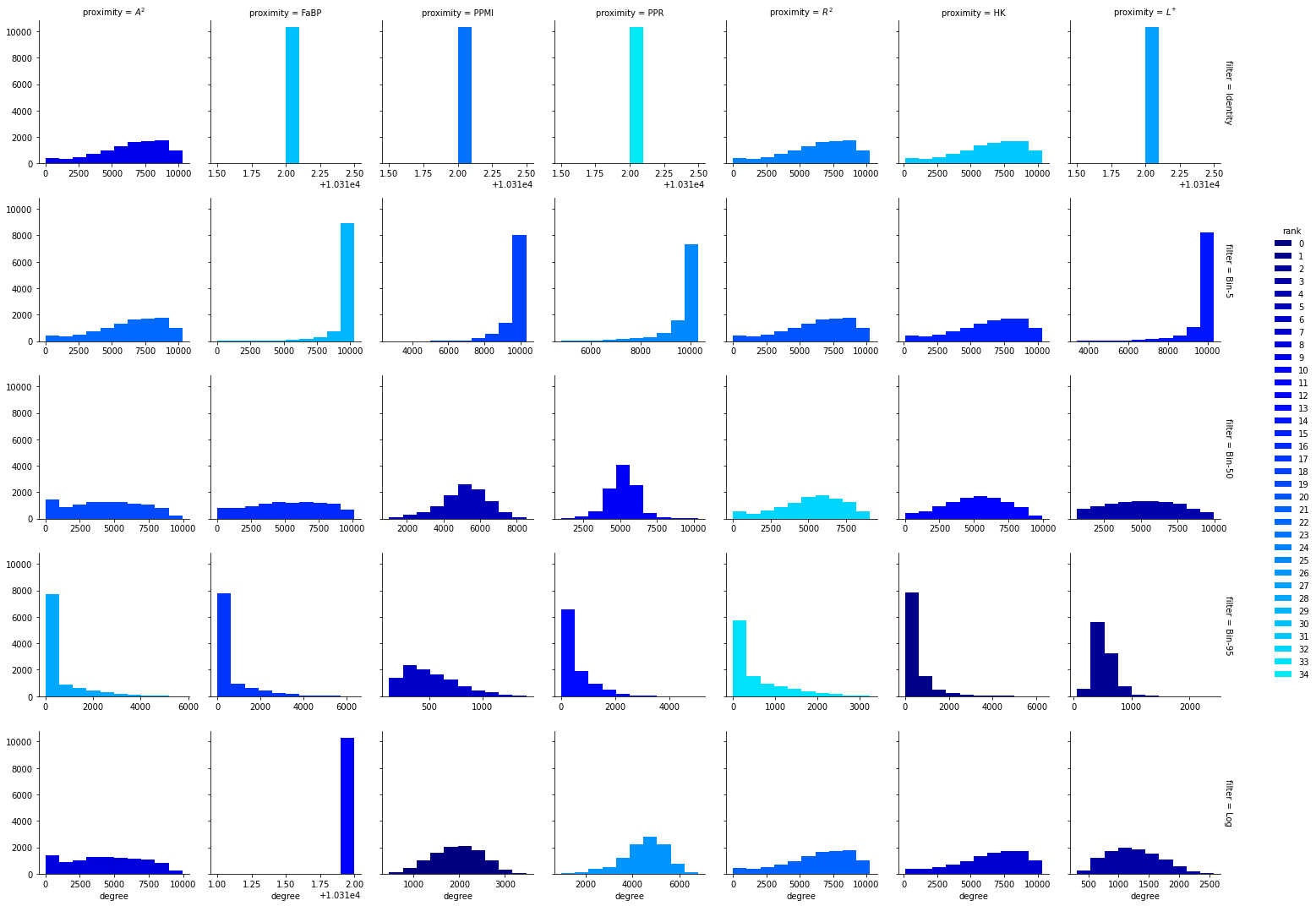}
    \caption{Degree (row sum) distributions of matrices resulting from all different combinations of $\proximity{}$ and $\nonlinearity{}$ on BlogCatalog.  In general, some of the best-performing embedding methods (darker colors) come from matrices whose row sum distributions follow a bell curve rather than a power law.}
    \label{fig:prox-degdist}
\end{figure*}

Many successful design choices produce a bell shape distribution of row sums.  For example, the \rampedLog nonlinearity filter (the best-performing nonlinearity on average) produces bell-shaped row sum distributions for all proximity matrices. \bin{95} produces a somewhat bell-shaped distribution for \ppmi, one of the best-performing proximity matrices.  In general, this lines up with our intuition to expect mostly moderate row sums.  %

\begin{observation}
\label{obs:prox-properties}
Some of the matrices yielding the best positional node embeddings have a bell curve rather than a power law distribution of row sums: that is, most nodes have moderate total proximity scores to all other nodes. \\
\end{observation} 
\subsection{Structural Embeddings}
\label{app:struc-stat-dist}
\begin{figure*}[h!]
    \centering
    \includegraphics[width=\linewidth]{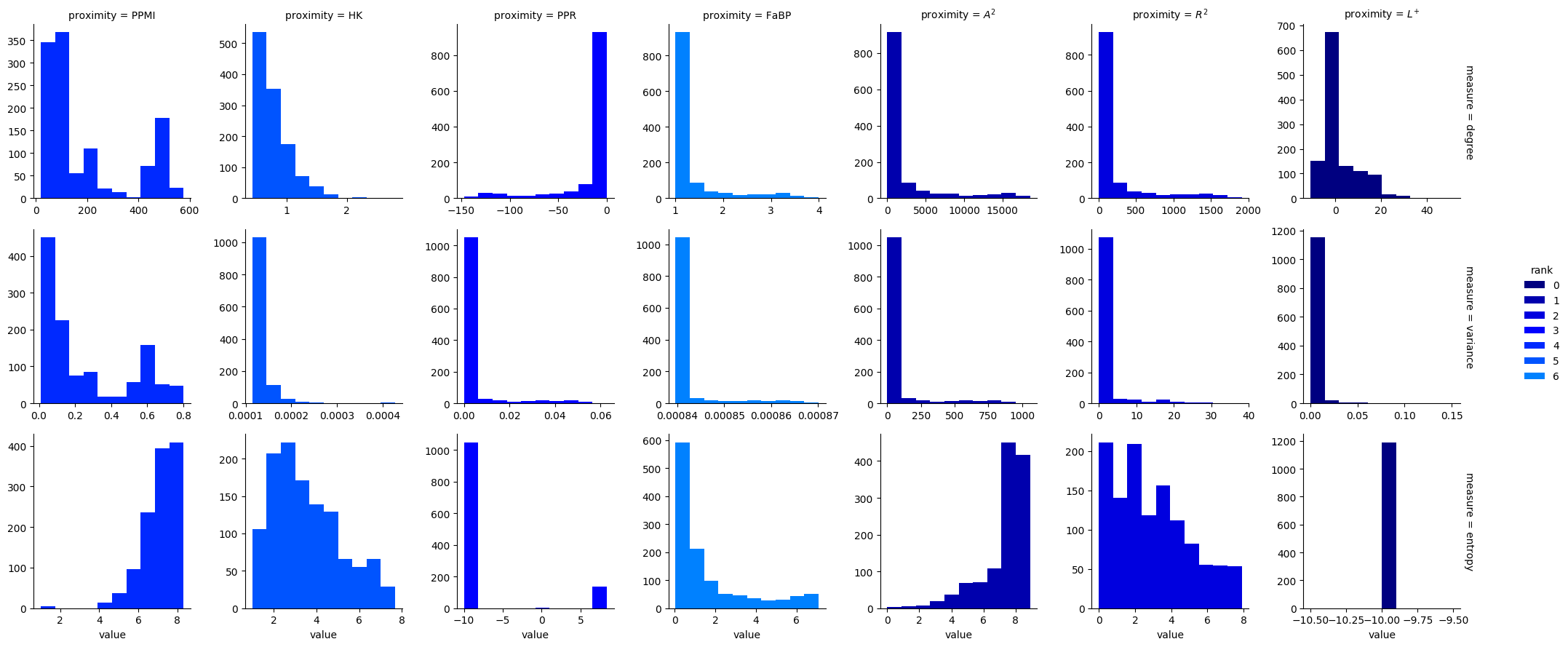}
    \caption{Distribution of row distribution statistics (degree, variance, and entropy) of proximity matrices (without nonlinearity) on USA Airports dataset used for structural embeddings.  Methods leading to more accurate structural embeddings have darker color.  Some of the best embeddings come from matrices whose rows' variance and entropy follow a power law distribution.}
    \label{fig:struct-dist-stats}
\end{figure*}

The CFS embedding method treats each row of the proximity matrix as a probability distribution. When learning structural embeddings using CFS, GraphWave notes that two cases will be uninformative for structural embedding: if the distribution of row entries is either too uniform or if it has too few nonzero values.

\noindent \textbf{Setup.}
To measure the row-wise uniformity of the matrices, we consider the variance of each row.  Meanwhile, we use entropy to diagnose rows with a few large entries and otherwise mostly small ones.   Thus, we plot the row-wise distribution of variances and entropy for each proximity matrix, as well as the distribution of row sums as in \S~\ref{app:prox-stat-dist}.  For brevity, we do not consider nonlinearity.  Note that for \ppr and $\invlap$, we truncate entropy values close to $-\infty$ to $-10$ for ease of visualization. 

\noindent \textbf{Results}.  
For most proximity matrices, the row-wise distribution of all three statistics tends to follow a power law distribution.  The row-wise sum and variance distributions for the \ppmi matrix, which generally leads to some of the the weaker structural embedding methods, tend to follow this pattern much more noisily.  Proximity matrices \ppr and \fabp, on the other hand, tend to follow an extreme power law distribution with a very thin tail.  This may indicate that a moderate power law distribution may be the most informative for structural embeddings.  The contrast with proximity-preserving embeddings (\S~\ref{app:prox-stat-dist}), where the most successful embeddings tended to come from matrices with a bell-curve distribution of row sums, corroborates our finding that the best positional and structural node embedding methods tended to use very different design choices. 

\begin{observation}
\label{obs:struc-properties}
Proximity matrices that yield good structural embeddings often follow a power law distribution of row-wise statistics.  The contrast with Obs.~\ref{obs:prox-properties} indicates that proximity matrices that lead to successful positional or structural node embeddings may have fundamentally different properties.  
\end{observation}

\end{document}